\documentclass[11pt]{article}
\usepackage{amsmath,amsfonts,amscd,amsthm,amssymb,amstext}
\usepackage{vaucanson-g}
\usepackage{xspace}
\usepackage[loose]{subfigure}
\theoremstyle{plain}
\newtheorem{theorem}{Theorem}
\newtheorem{lemma}[theorem]{Lemma}
\newtheorem{corollary}[theorem]{Corollary}
\newtheorem{proposition}[theorem]{Proposition}

\newtheorem{conjecture}[theorem]{Conjecture}
\newtheorem{definition}{Definition}
\theoremstyle{definition}
\newtheorem{example}{Example}
\newtheorem{remark}{Remark}

\newcommand{\theor}[1]{Theorem~\ref{t.#1}}
\newcommand{\propo}[1]{Proposition~\ref{p.#1}}
\newcommand{\corol}[1]{Corollary~\ref{c.#1}}

\newcommand{\remar}[1]{Remark~\ref{r.#1}}

\newcommand{\lemme}[1]{Lemma~\ref{l.#1}}
\newcommand{\figur}[1]{Figure~\ref{f.#1}}

\newcommand{\equat}[1]{Equation~(\ref{q.#1})}
\newcommand{\equnm}[1]{(\ref{q.#1})}
\newcommand{\exemp}[1]{Example~\ref{e.#1}}
\newcommand{\EoP}{\hbox{}\hfill\qedsymbol\hbox{}}%
{\end{example}}
{\end{example}}
\newcommand{\slparagraph}[1]%
    {\par\noindent{\textit{#1}}\hspace*{.8em}\ignorespaces}
\newcommand{\jsListe}[1]%
    {\par\noindent
     \makebox[\parindent][r]{{\rm (#1)}}\hspace*{.8em}\ignorespaces}

\newcommand{\jsListb}[1]%
    {\noindent
     \makebox[2.2\parindent][r]{{\rm #1)}}\hspace*{0.8em}\ignorespaces}

\newcommand{\fa}{\forall}

\newcommand{\bk}{\setminus}
\newcommand{\UL}[1]{\underline{#1}}
\newcommand{\msp}{\hspace*{0.2em}} % espace pour faire ressortir
\newcommand{\xmd}{\hspace{0.125em}} % espace entre les symboles
\newcommand{\e}{\text{\quad}}                 % un moins petit espace
\newcommand{\ee}{\text{\qquad}}               % un espace
\newcommand{\eee}{\text{\qquad \qquad}} % et un grand
\renewcommand\leq\leqslant
\renewcommand\geq\geqslant
\renewcommand\epsilon\varepsilon
\newcommand{\eqpnt}{\makebox[0pt][l]{\: .}}

\newcommand{\eqvrg}{\makebox[0pt][l]{\: ,}}
\newcommand{\EqVrgInt}{\: , \e }

\newcommand{\quantvrg}{\, , \;}
\newcommand{\quantsp}{\ee }
\newcommand{\quantsmsp}{\e }

\newcommand{\cf}{\textit{cf.}\xspace}
\newcommand{\ipsofacto}{\textit{ipso facto}\xspace}
\newcommand{\via}{\textit{via}\xspace}
\newcommand{\ExtnF}[1]%
   {\overset{{\scriptscriptstyle \pmb{\smile}}}{#1}}

\newcommand{\DiffF}[1]%
   {\overset{{\scriptscriptstyle \pmb{\lor}}}{#1}}
\newcommand{\LocaF}[1]%
   {\overset{{\scriptscriptstyle \leftrightarrow}}{#1}}

\newcommand{\jsDist}[2][{}]%
   {\operatorname{\mathbf{d}_{#1}}\left(#2\right)}

\newcommand{\supp}{\operatorname{{\mathsf{supp}\,}}}

\renewcommand{\lim}{{\operatornamewithlimits{\mathsf{lim}}}}

\newcommand{\ETAze}[1]{0_{#1}}

\newcommand{\zeK}{\ETAze{\K}}

\newcommand{\x}{\! \times \!}
\newcommand{\SerSAnMon}[2]%
    {#1 \langle \! \langle  #2  \rangle \! \rangle }
\newcommand{\SerSAnMonD}[2]%
    {\left[#1\right] \langle \! \langle  #2  \rangle \! \rangle }
\newcommand{\SerMon}[1]%
    {\!\langle \! \langle  #1  \rangle \! \rangle }
\newcommand{\KAe}{\SerSAnMon{\K}{\Ae }}

\newcommand{\KA}{\KAe} % compatibilit\'e

\newcommand{\PolSAnMon}[2]%
    {{#1 \langle  #2 \rangle }}
\newcommand{\PolMon}[1]%
    {{\!\langle  #1 \rangle }}

\newsavebox{\LeftBraket}
\savebox{\LeftBraket}{\scalebox{0.7 1.2}{$<$}}
\newsavebox{\RightBraket}
\savebox{\RightBraket}{\scalebox{0.7 1.2}{$>$}}
\newcommand{\bra}[1]{\hbox{}\usebox{\LeftBraket}%
                           #1\usebox{\RightBraket}\hbox{}}
\newcommand{\Rat}{\mathrm{Rat}\,}
\newcommand{\NRat}{\N \Rat}

\newcommand{\lmn}{{(\lambda , \mu , \nu )}}

\newcommand{\lmnu}{{(\lambda _1, \mu _1, \nu _1)}}

\newcommand{\ekz}{{(\eta ,\kappa ,\zeta )}}

\renewcommand{\Alph}{A}
\newcommand{\Blph}{B}

\newcommand{\jsStar}[1]{{{#1}^{*}}}
\newcommand{\Ae}{\jsStar{\Alph}}
\newcommand{\Be}{\jsStar{\Blph}}

\newcommand{\iotaK}{\iota_{\ShiftInd{K}}}

\newcommand{\unAe}{{1_{\Ae}}}

\newcommand{\compos}{\ccdot }
\newcommand{\matmul}{\mathbin{\cdot}}

\newcommand{\hadam}{\mathbin{\odot }}

\newcommand{\phiikpsi}%
{{\varphi ^{-1}\! \compos        \iotaK \! \compos \! \psi }}
\newcommand{\phiiotpsi}[1]%
{{\varphi ^{-1}\! \compos        \iota _{\ShiftInd{#1}} \! \compos \! \psi }}
\newcommand{\phiintkpsi}[1]%
{{(#1\varphi ^{-1}\! \cap K) \psi }}

\newcommand{\jsgeq}{\geqslant }
\newcommand{\jsless}% 001102
   {\mathrel{\leqslant_{\!\!\!\!\scriptscriptstyle{/}}}}
\newcommand{\jsgrea}% 020120
   {\mathrel{\geqslant_{\!\!\!\!\scriptscriptstyle{\backslash}}}}

\newcommand{\lexiconeq}% 010125
   {\preccurlyeq_{\!\!\!\!\!\scalebox{1.8 1}{\scriptscriptstyle{\pmb{/}}}}}
\newcommand{\jsAutUn}[1]%  
   {\mbox{$\left\langle \thinspace #1 \thinspace \right\rangle $}}
\newcommand{\Tran}[1]{\bigl ( #1\bigr )}
\newcommand{\Acu}{{{\Ac_{1}}}}
\newcommand{\ShiftInd}[1]{\raisebox{-0.3ex}{$\scriptstyle{#1}$}}
\newcommand{\actb}{\mathbin{\raisebox{0.2ex}%
                        {${\scriptscriptstyle \circ} $}}}
\newcommand{\ccdot}{\actb} % pour compatibilit\'e
\newcommand{\bornedeuxlignes}[2]%
{\mbox{$
\begin{array}{c}{\scriptstyle #1}\\ {\scriptstyle #2} \end{array}
       $}}
\newcommand{\path}[1]{\xrightarrow{\ #1 \ }} %100309 pb
\newcommand{\pathaut}[2]{\underset{#2}{\path{#1}}}
\newcommand{\ExpDer}[2][a]%
   {{\operatorname{\frac{\partial}{\partial \mbox{$#1$}}}#2}}
\newcommand{\ExpDerr}[2][a]%
    {\operatorname{\frac{\partial_{\mathrm{R}}}{\partial \mbox{$#1$}}}#2}
\newcommand{\ExpDerB}[2][a]% 020429
   {{\operatorname{\frac{\partial '}{\partial '\mbox{$#1$}}}#2}}

\newcommand{\comment}[1]{}
\newcommand{\radix}{\prec}
\newcommand{\card}[1]{\mathrm{card} \left( {#1} \right)}
\newcommand{\lgt}[1]{|#1|}
\newcommand{\K}{\mathbb{K}}
\newcommand{\N}{\mathbb{N}}
\newcommand{\Z}{\mathbb{Z}}
\newcommand{\Ac}{\mathcal{A}}

\newcommand{\Cc}{\mathcal{C}}
\newcommand{\Dc}{\mathcal{D}}

\newcommand{\Sc}{\mathcal{S}}

\newcommand{\Ak}{A_{k}}
\newcommand{\muf}[1]{\mu\left(#1\right)}

\newcommand{\kappaf}[1]{\kappa\left(#1\right)}
\newcommand{\Repr}[2]{\langle#2\rangle_{\ShiftInd{#1}}}

\newcommand{\RU}[1]{\Repr{U}{#1}}
\newcommand{\RL}[1]{\Repr{L}{#1}}
\newcommand{\RLu}[1]{\Repr{L_{1}}{#1}}
\newcommand{\RSc}[1]{\Repr{\Sc}{#1}}
\newcommand{\Ev}{\pi}

\newcommand{\Eval}[2]{\mathop{\Ev_{\ShiftInd{#1}}}\left(#2\right)}

\newcommand{\EvSc}[1]{\Eval{\Sc}{#1}}
\newcommand{\EvL}[1]{\Eval{L}{#1}}
\newcommand{\EvLu}[1]{\Eval{L_{1}}{#1}}

\newcommand{\Enm}{\mathbf{E}}
\newcommand{\Enum}[1]{\Enm_{\ShiftInd{#1}}}
\newcommand{\EnSc}{\Enum{\Sc}}
\newcommand{\EnL}{\Enum{L}}
\newcommand{\Pof}[1]{P(#1)}
\newcommand{\Pu}{\Pof{u}}
\newcommand{\Pua}{\Pof{u\xmd a}}
\newcommand{\Lu}{L_{1}}

\newcommand{\Succ}[1]{\mathsf{Succ} _{#1}}
\newcommand{\SuccL}{\Succ{L}}

\SmallPicture 
\newlength{\lga}\newlength{\lgb}
\begin{document}
\title{On the enumerating series \\ of an abstract numeration system}
\author{Pierre-Yves Angrand\thanks{
LTCI (UMR 5141), Telecom ParisTech,
46 rue Barrault, 75634 Paris Cedex 13, France,
{\small{\tt angrand@enst.fr}} --- Corresponding author.}
\and Jacques Sakarovitch\thanks{
LTCI (UMR 5141), CNRS / Telecom ParisTech,
% 46 rue Barrault, 75634 Paris Cedex 13, France
{\small{\tt sakarovitch@enst.fr}}.}
}
%\keywords{rational series, abstract numeration system}
% \subjclass{68Q45,68Q70}
\date{}
\maketitle

\begin{abstract}
It is known that any rational abstract numeration system is 
faithfully, and effectively, represented by an $\N$-rational series. 
A simple proof of this result is given which yields a representation 
of this series which in turn allows a simple computation of the value 
of words in this system and easy constructions for the recognition of 
recognisable sets of numbers.

It is also shown that conversely it is 
decidable whether an $\N$-rational series corresponds to a rational 
abstract numeration system. 
\end{abstract}

\section{Introduction}

In order to state our result, we have first to recall the definition 
--- due to Lecomte and Rigo \cite{LecoRigo01} --- of an 
\emph{abstract numeration system} and, in order to motivate it, the more 
common one of numeration systems.

Numbers do exist independently of the way we represent them, and 
operations on numbers are defined independently of the way they are 
computed.
The role of a numeration system is to set a framework in which numbers 
are represented by \emph{words} (over a suitable alphabet) allowing 
to describe operations on numbers as algorithms on the 
representations, that is, on words.

The most common numeration system --- in our modern times ---
is the $k$-ary system where numbers are given their representation 
\emph{in base~$k$}, that is, written as words over the alphabet
$\msp \Ak = \{0, 1, \ldots, k-1\}\msp$ and which do not start 
with~$0$ (but for the representation of~$0$ itself).
The sequence of the representations of the integers in the binary 
system is:
$\msp
\{0, 1, 10, 11, 100, 101, 110, \ldots\}
\msp$.

While keeping the notion of position numeration system, the $k$-ary 
systems can be generalised by replacing the 
sequence
$\msp(k^{n})_{n\jsgeq 0}\msp$ 
with some increasing sequence
$\msp U=(U_{n})_{n\jsgeq 0}\msp$ of integers such that
$\msp U_{0}=1\msp$.
Using a greedy algorithm, every integer~$n$ is then given a 
representation in the `base'~$U$, called its $U$-representation and 
denoted by~$\RU{n}$.
A well-known example is the Fibonacci numeration system based on the 
sequence
$\msp F=(F_{n})_{n\jsgeq 0}\msp$
of Fibonacci numbers starting with
$\msp F_{0}=1\msp$ and $\msp F_{1}=2\msp$.
In this system, every positive integer is given a \emph{canonical} 
representation which is computed by the greedy algorithm and which is 
characterised by the fact it does not contain~$11$ as a factor.
The sequence of the representations of the integers in the Fibonacci 
system is:
$\msp
\{0, 1, 10, 100, 101, 1000, 1001, 1010, \ldots\}
\msp$.

It is possible to look at these two numeration systems, the $2$-ary 
system and the Fibonacci system, independently from the sequences
$\msp(2^{n})_{n\jsgeq 0}\msp$
and 
$\msp (F_{n})_{n\jsgeq 0}\msp$
and the greedy algorithm, 
and by just considering \emph{the set of words} that represent the 
integers:
$\msp 1\{0,1\}^{*}\cup\{0\} \msp$
in the first case,
$\msp 1\{0,1\}^{*}\bk \{0,1\}^{*}11\{0,1\}^{*} \cup \{0\} \msp$
in the second case
and by \emph{enumerating the element of this set} in the radix 
order.\footnote{%
    The definition of radix order will be given below.}
In both cases, every integer will be given the same representation 
without reference to the way this representation is computed.
It is the language of all representations that matters and this
naturally leads to the definition of abstract numeration 
systems.

\begin{definition}[\cite{LecoRigo01}]
\label{d.ans}
An \emph{abstract numeration system}
(or \emph{ANS} for short) is a triple $\Sc=(L,A,<)$ 
where~$A$ is an alphabet equipped with a total order~$<$ and~$L$ 
is an infinite language of~$\Ae$.

The system~$\Sc$ allows to define a one-to-one correspondence 
between~$\N$ and~$L$ by associating every integer~$n$ with the 
$(n+1)$-th word of~$L$ in the radix order defined on~$\Ae$ by~$<$.
This \emph{representation} of~$n$ is denoted by~$\RSc{n}$ and 
conversely the corresponding \emph{value} of a word~$w$ of~$L$ is 
denoted by~$\EvSc{w}$.
Of course, the following holds:
\begin{equation}
    \RSc{\EvSc{w}}= w \ee \text{and}\ee \EvSc{\RSc{n}} = n
\eqpnt 
\notag
%       \label{}
\end{equation}
In most cases, the alphabet~$A$ and the order~$<$ on~$A$ are 
fixed and understood and we speak of the ANS defined by the 
language~$L$ and we use the simpler notations~$\RL{n}$ and~$\EvL{w}$. 

If~$L$ is a rational language of~$\Ae$, we say that the ANS is 
rational.    
\end{definition}

\begin{example}
    \label{e.eve}
    Let
    $\msp A= \{a,b\}\msp$, with
    $\msp a<b \msp$ and let~$\Lu$ be the language of words with an 
    even number of $b$'s:
    $\msp \Lu = 
       \left\{ w\in\{a,b\}^{*}\mid \lgt{w}_{b}\equiv 0\mod 2\right\}\msp$.
    The sequence of the representations of the integers is:
    $\msp\{\epsilon, a, aa, bb, aaa, abb, bab, \ldots\}\msp$ and, 
    for instance,
    \begin{equation}
        \RLu{18}= a\xmd a\xmd b\xmd a\xmd b 
        \ee \text{and}\ee 
        \EvLu{b\xmd b\xmd a\xmd b\xmd b} = 29
    \eqpnt 
    \notag
    %       \label{}
    \end{equation}
\end{example}

Beyond the irrepressible appeal to generalisation and abstraction, a 
true motivation that supports the definition of ANS is 
to understand 
which properties of a numeration system depend upon the whole
language of the representations only, and which are 
more directly related to the way the representation of 
every number is computed.
For instance, we have shown in a previous paper~\cite{AngrSaka10} 
that the successor function in a  
rational ANS is a piecewise cosequential function, 
whereas the characterisation of those 
systems for which this successor function is co-sequential is known 
in the case of $\beta$-numeration systems (\cf~\cite{Frou97}) but 
seems to be out of reach for arbitrary rational ANS so far.

The purpose of this paper is to set up even tighter bonds between 
rational abstract numeration systems and classical automata theory.
We reach this goal \via the definition of the \emph{enumerating 
series} of a numeration system and with the use of its \emph{representation} 
in the case it is rational.

\begin{definition}
    \label{d.enu-ser}
    Let $\Sc=(L,A,<)$  be an {abstract numeration system}.
    The enumerating series of~$\Sc$ is the $\N$-series over~$\Ae$ 
    denoted by~$\EnSc$ and defined by:
    \begin{equation}
        \EnSc= \sum_{w\in L} \left(\EvSc{w}+1\right)\xmd w
	\eqpnt 
	\notag
%         \label{}
    \end{equation}
As above, the notation can be simplified as
$\msp\EnL= \sum_{w\in L} \left(\EvL{w}+1\right)\xmd w\msp$.
\end{definition}

\begin{remark}
    \label{r.enu-ser}
    The above definition has been taken so that the language~$L$ is 
    entirely determined by~$\EnL$. Indeed, 
    \begin{equation}
        L = \supp(\EnL)
	\eqpnt 
	\notag
%         \label{}
    \end{equation}
One certainly could have taken
$\msp\EnL= \sum_{w\in L} \left(\EvL{w}\right)\xmd w\msp$
as a definition for~$\EnL$. 
All the results we are going to describe would have been valid and it 
may have looked more natural. 
But we would have lost the information on the first word of~$L$, that 
is, the representation of~$0$.
\end{remark}

The starting point of our work is a direct proof of the following 
result (the definition of $\N$-rational series will be recalled below).

\begin{theorem}[\cite{ChofGold95}]
    \label{t.enu-ser-rat}%
    The enumerating series of a rational abstract numeration system 
    is an $\N$-rational series. 
\end{theorem}

In~\cite{ChofGold95}, \theor{enu-ser-rat} was a corollary of 
constructions set up for 
establishing the rationality or algebraicity of a family of counting 
problems by means of rational transductions.
\theor{enu-ser-rat} was also given another and specialised proof 
in~\cite{Rigo01}. 
Even if both this and the original proofs are effective, the one we 
give below in Section~3 amounts to compute directly 
a \emph{representation} of~$\EnL$ 
from a representation of (the characteristic series of)~$L$
and also to give an even more compact algorithm for 
calculating the coefficient of a word~$w$ in~$\EnL$, that is, the 
value of~$w$ in the system~$L$ increased by~1.
We then deduce from this latter algorithm the construction of the 
automaton that recognises the set of  
representations in the system~$L$ of a recognisable set of numbers 
(Section~4). 
It is to be noted that the same last construction was also given 
in~\cite{KrieEtAl09} (\cf \remar{ref-dec}).

The next result plays the role of a converse of \theor{enu-ser-rat}: 
of course not every $\N$-rational series is the enumerating series of 
a rational number system, but one can at least know when it is the 
case.

\begin{theorem}
    \label{t.enu-ser-dec}%
    It is decidable whether an $\N$-rational series is the enumerating 
    series of a (rational) abstract numeration system or not.
\end{theorem}

We end the paper with some problems that are directly inspired by 
\theor{enu-ser-rat}.

\section{Preliminary and notation}

This paper makes use of several notions of automata theory such as 
unambiguous, or deterministic, automata, rational series and 
languages, with which the reader is supposed to be familiar. 
Definitions that are not given here are to be found in 
reference books such as~\cite{Eile74,BersReut88,Saka03}.
Our notation are mainly those used in~\cite{Saka03}.

In the sequel, $A$ is a finite alphabet, $\Ae$ the free monoid
generated by~$A$, $\unAe$ the empty word, identity of~$\Ae$.
The length of a  word~$w$ in~$\Ae$ is denoted by~$|w|$.
Let~$A$ be \emph{totally ordered} by~$<$.  
The \emph{radix order}~$\radix$ on~$\Ae$ %is the order whose strict part
is defined by:\footnote{%
   Notice that $\radix$ is not reflexive and is not the order but the 
   \emph{strict part} of the radix order.}
\begin{equation}
    u \radix v \ee\text{if}\e
    \left\{
    \begin{array}{ll}
        \text{either} & |u| < |v| \eqvrg   \\
        \text{or}     & |u| = |v| \EqVrgInt
                        u=w\xmd a\xmd u'\EqVrgInt
                        v=w\xmd b\xmd v' \e\text{and}\e
                        a < b \eqpnt
    \end{array}
    \right.
    \notag
%     \label{}
\end{equation}
The radix order is a \emph{well order}, that is, every non empty subset
of~$\Ae$ has a smallest element for~$\radix$ and can thus be used to 
\emph{enumerate} any subset of~$\Ae$.

Let~$\K$ be a semiring; for instance, $\N$, the semiring of non 
negative integers.
A ($\K$-)\emph{series}~$s$ (over~$\Ae$) is a 
\emph{map} from~$\Ae$ to~$\K$, and the image of a word~$w$ by~$s$ is 
called the \emph{coefficient of~$w$ in~$s$} and is denoted 
by~$\bra{s,w}$. 
The set of series over~$\Ae$ with coefficients in~$\K$ is denoted 
by~$\KA$.
The \emph{support} of a series~$s$ is the language, denoted 
by~$\supp s$, which contains those words whose 
coefficient in~$s$ is different from~$\zeK$.
Conversely, the \emph{characteristic series} of a language~$L$ 
of~$\Ae$ is the $\N$-series, denoted by~$\UL{L}$, defined by
$\bra{\UL{L},w}=1$ if~$w$ is in~$L$ and  $\bra{\UL{L},w}=0$ otherwise.

We call ($\K$-)\emph{representation}, of \emph{dimension}~$n$, a 
triple~$\lmn$ where~$\mu$ is a morphism
$\msp\mu\colon\Ae\rightarrow\K^{n\x n}\msp$ from~$\Ae$ to the 
$n\x n$-matrices with entries in~$\K$,
and~$\lambda$ and~$\nu$ are two vectors of dimension~$n$ with entries 
in~$\K$, $\lambda$ a row vector and~$\nu$ a column vector.
A series~$s$ in~$\KAe$ is ($\K$-)\emph{recognisable} if there exists a 
representation~$\lmn$ such that, for every~$w$ in~$\Ae$,
\begin{equation}
    \bra{s,w} = \lambda\matmul\muf{w}\matmul\nu
    \eqpnt
    \notag
\end{equation}

A series is ($\K$-)\emph{rational} if it is the behaviour of a finite 
($\K$-)automaton, that is, an automaton with multiplicity in~$\K$
(the behaviour of an automaton~$\Ac$ is the series where the 
coefficient of a word~$w$ is the sum of the multiplicities of all 
computations in~$\Ac$ with label~$w$).
Finite $\K$-automata whose transitions are labelled by 
\emph{letters} and $\K$-representations are two ways to describe the 
same concept.\footnote{%
   This is true only because~$\Ae$ is a \emph{free monoid}.}
The illustration given with \exemp{eve} suffices for the definition.
As every $\K$-automaton is equivalent to one which is labelled by 
letters, the  
families of $\K$-rational and $\K$-recognisable series coincide.

\begin{example} %, with $a<b$
    \label{e.eve-2}
The language~$\Lu$ of words with an even number
of~$b$'s is recognised by the automaton~$\Acu$ drawn 
at~\figur{rec-num-1}. 
The representation~$\lmnu$ associated with~$\Acu$ is
\begin{equation}
\lambda_{1} = \begin{pmatrix}1 & 0\end{pmatrix}
\EqVrgInt 
\mu_{1} (a) = 
\begin{pmatrix}
    1 & 0 \\ 0 & 1
\end{pmatrix}
\EqVrgInt 
\mu_{1} (b) = 
\begin{pmatrix}
    0 & 1 \\ 1 & 0
\end{pmatrix}
\EqVrgInt 
\nu_{1} = 
\begin{pmatrix}
    1 \\  0
\end{pmatrix}
\eqpnt 
\notag
\end{equation}
\end{example}

In the sequel, we mostly use~$\N$ as the semiring, and we may call 
\emph{representation} an $\N$-representation.

\setlength{\lga}{4cm}
\begin{figure}[ht]
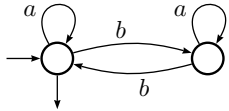

\centering
\VCDraw{%[q_1][q_2]
% \begin{VCPicture}{(-5,-1)(5,2)}
% \State{(-0.5\lga,0)}{A}
% \State{(0.5\lga,0)}{B}
\begin{VCPicture}{(0,-1)(\lga,1.5)}
%     \SmallState[.5][.5]
\State{(0,0)}{A}\State{(\lga,0)}{B}
\Initial{A}\Final[s]{A} 
\ArcL{B}{A}{b}\ArcL{A}{B}{b}
\LoopN{B}{a}\LoopN{A}{a}
\end{VCPicture}}
\caption{A \emph{DFA} accepting words with an even number of $b$'s}
\label{f.rec-num-1}
\end{figure}
\setlength{\lga}{3cm}

\section{Representation of the enumerating series}

The proof of \theor{enu-ser-rat}, as given in~\cite{ChofGold95} where 
\theor{enu-ser-rat} is (part of) Corollary~8, is based on the  
construction of an unambiguous rational transduction that associates 
to every word~$u$ all words~$v$ that are greater than~$u$ in the 
radix order.
From this, it is easy to derive that the image of the characteristic 
series of a rational language~$L$ is a recognisable series, and equal 
to~$\EnL$ --- up to the intersection (or Hadamard product) with~$L$.
The advantage of this construction is that it can be applied to 
unambiguous context-free languages and to various other counting 
functions as well.
The inconvenient is that it does not provide directly the 
representation of~$\EnL$, although it is very similar to the one we 
develop below.

In~\cite{Rigo01}, \theor{enu-ser-rat} is Proposition~29; its proof is 
more direct than in~\cite{ChofGold95} in the sense it does not rely 
on the rational transduction machinery but makes use instead of the 
characterisation of recognisable series as those which belong to a 
finitely generated stable submodule of~$\KA$.
But this proof yields neither the representation of~$\EnL$ nor a 
simple mean to compute it.

\subsection{Preparation}

If~$a$ is a letter of~$A$, 
let us denote by~$A_a$ the set of letters of~$A$ smaller than~$a$:
\begin{equation}
    A_a = \{b\in A\mid b<a\}
\eqpnt
\notag 
\end{equation}
If~$u$ be a word of~$\Ae$, let us denote by~$\Pu$ the set 
of words of~$\Ae$ (strictly) smaller than~$u$ in the radix order:
\begin{equation}
\Pu = \{v\in\Ae\mid v \radix u\}
\eqpnt
\notag 
\end{equation}
This set~$\Pu$ can be defined by induction on the length of~$u$ by 
the following remark.
Any word smaller than~$u$ followed by \emph{any} letter is smaller 
than~$u\xmd a$, and so is~$u\xmd b$ for any letter~$b$ smaller 
than~$a$, and the empty word is also smaller than~$u\xmd a$.
These three sets are pairwise disjoint and any word smaller 
than~$u\xmd a$ falls in one of them.
Altogether, we have proved the following lemma:\footnote{%
   \cf \remar{ref-dec} below.}
\begin{lemma}
    \label{l.enu-ser-rat}
\ee    
$\msp \fa u\in\Ae\quantvrg \fa a \in A \quantsp
\Pua = \unAe\cup u A_a  \cup \Pu A \msp$.
% 
% \eqpnt
% \eee
% \label{q.dec-wor}
% \end{equation}
% \begin{equation}
%     \fa u\in\Ae\quantvrg \fa a \in A \quantsp
% \Pua = \unAe\cup u A_a  \cup \Pu A \eqpnt
% \eee
% \label{q.dec-wor}
% \end{equation}
\end{lemma}

Let~$L$ be a rational language of~$\Ae$ and~$\lmn$ the 
$\N$-representation which  
corresponds to an \emph{unambiguous} finite automaton which 
recognises~$L$:
\begin{equation}
    \fa w\in\Ae \quantsp
\lambda\matmul\muf{w}\matmul\nu = 1
\e\Longleftrightarrow\e
w\in L
\eqpnt
\eee
\notag
% \label{q.}
\end{equation}
We use the following notation: if~$K$ is a (finite) subset of~$\Ae$, 
then
$\msp\muf{K}=\sum_{w\in K}\muf{w}\msp$.
As~$\lmn$ corresponds to an \emph{unambiguous} automaton, we have:
\begin{equation}
    \fa K\subseteq\Ae \quantsp
\lambda\matmul\muf{K}\matmul\nu = 
\sum_{w\in K}\lambda\matmul\muf{w}\matmul\nu=
\card{K\cap L}
\eqpnt
\eee
% \notag
\label{q.cnt-rep}
\end{equation}

\subsection{Proof of \theor{enu-ser-rat}}
    
Let $\Sc=(L,A,<)$ be a rational ANS, 
$\Ac$ an unambiguous automaton that recognises~$L$, 
$\lmn$ the corresponding $\N$-representation, and~$k$ its dimension.
From~\equnm{cnt-rep} follows:
%     
% In particular, we have:
\begin{equation}
    \fa u\in\Ae \quantsp
\lambda\matmul\muf{\Pu}\matmul\nu = 
\card{\{v\in\Ae\mid v\in L \e\text{and}\e v \radix u\}}
\eqpnt
\eee
\notag
% \label{q.}
\end{equation}
and thus:
\begin{equation}
    \fa w\in L \quantsp
\lambda\matmul\muf{\Pof{w}}\matmul\nu =
\EvL{w}
\eqpnt
\eee
\notag
% \label{q.}
\end{equation}
% Let us first establish that~$\lambda\matmul\muf{\Pu}\matmul\nu$ is 
% computed by an $\N$-representation.
From \lemme{enu-ser-rat} follows:\footnote{%
   \cf \remar{ref-dec} below.}
\begin{multline}
    \fa u\in\Ae\quantvrg \fa a \in A \quantsp\\
\lambda\matmul\muf{\Pua}\matmul\nu = 
                \lambda\matmul\muf{\unAe}\matmul\nu
              + \lambda\matmul\muf{u}\matmul\muf{A_a}\matmul\nu \ee\\ 
              + \lambda\matmul\muf{\Pu}\matmul\muf{A}\matmul\nu 
\eqpnt\ee
% \notag
\label{q.pua-mu-1}
\end{multline}
Let
$\sigma=\muf{A}$ and, for every~$a$ in~$A$,
$\sigma_{a} = \muf{A_{a}}$.
Thus~\equnm{pua-mu-1}~is rewritten as:
\begin{multline}
    \fa u\in\Ae\quantvrg \fa a \in A \quantsp\\
\lambda\matmul\muf{\Pua}\matmul\nu = 
                \lambda\matmul\nu
              + \lambda\matmul\muf{u}\matmul\sigma_{a}\matmul\nu 
              + \lambda\matmul\muf{\Pu}\matmul\sigma\matmul\nu 
\eqpnt\ee
% \notag
\label{q.pua-mu-2}
\end{multline}
Let~$\ekz$  be the representation of dimension~$2k+1$ described by the 
following $(1,k,k)$-block decomposition:
\begin{equation}
\eta=
\begin{pmatrix}1 & \lambda& 0\end{pmatrix}
\EqVrgInt
\fa a\in A \quantsmsp
\kappaf{a}=
\begin{pmatrix}
    1 & 0& \lambda\\
    0 & \muf{a}& \sigma_{a}\\
    0 & 0& \sigma
\end{pmatrix}
\EqVrgInt
\zeta=
\begin{pmatrix}0 \\ 0 \\ \nu\end{pmatrix}
\eqpnt
\notag
% \label{q.}
\end{equation}
It is routine to verify, by induction on the length of~$u$, and based 
on \lemme{enu-ser-rat}, that
$\msp
\lambda\matmul\muf{\Pu}\matmul\nu =
\eta\matmul\kappaf{u}\matmul\zeta
\msp$ for every~$u$ in~$\Ae$.

Let now~$\xi=
\begin{pmatrix}1 \\ 0 \\ \nu\end{pmatrix}$ 
and let~$s$ be the series realised by~$(\eta,\kappa,\xi)$:
\begin{equation}
    \fa u\in\Ae \quantsp
    \bra{s,u} = 1 + \card{\{v\in\Ae\mid v\in L \e\text{and}\e v \radix u\}}
\eqpnt
\notag
% \label{q.}
\end{equation}
In order to get the enumerating series of~$L$, we must retain the 
words that belong to~$L$ only, that is, to make the Hadamard product with 
the \emph{characteristic series}~$\UL{L}$ of~$L$:
\begin{equation}
\EnSc = s \hadam \UL{L}
    \eqvrg
\notag
% \label{q.}
\end{equation}
and~$\EnSc$ is $\N$-rational as the Hadamard product of two 
$\N$-rational series (this is often referred to as (another) Sch\"utzenberger
Theorem\footnote{%
   \cf \cite[Th. I.5.3]{BersReut88}, \cite[Prop. VI.7.1]{Eile74} or
   \cite[Cor. III.3.9]{Saka03}.}).
% \end{proof}
\EoP

\begin{remark}
The construction underlying the proof  
yields for $\EnL$ an 
$\N$-representation of dimension~$2\xmd k^{2}+k$.
\end{remark}

\subsection{Computation of the value of a word}

The description of a~$\N$-rational series~$s$ by a~$\N$-representation gives
a way to compute the coefficient of any word~$w$ in~$s$.
\theor{enu-ser-rat} thus solves \ipsofacto the problem of computing 
the value~$\EvL{w}$ of a word~$w$ in a rational abstract number 
system~$L$ (which occupies the whole Sect.~2 in~\cite{LecoRigo10}).

If~$s$ has a representation~$(\chi,\omega,\phi)$ of dimension~$n$, and if~$w$ 
is of length~$\ell$, the general algorithm consists in computing
$\msp\chi\matmul\omega(w _{i+1})=
(\chi\matmul\omega(w_i))\matmul\omega(a _{i+1})\msp$
for $i=0$ to $i=\ell-1$, where~$a_i$ is the~$i$-th letter of~$w$ and~$w_i$ 
its prefix of length~$i$. 
Every step costs~$2\xmd n^2$ operations, thus in total,
roughly~$2\xmd\ell\xmd n^2$ operations. 

It would be  not such a good idea, however, to apply this 
general algorithm to the representation of dimension~$2k^2+k$ 
we have obtained in the proof of \theor{enu-ser-rat} above.
Its particular form allows, in fact, to compute
with vectors and matrix of dimension~$k$ only.

Given as above the unambiguous automaton~$\Ac$ of dimension~$k$ which
recognises~$L$ and the corresponding $\N$-representation~$\lmn$, we 
associate a pair $\Tran{\alpha(w),\gamma(w)}$ with every~$w$ in~$\Ae$, 
where~$\alpha(w)$ and~$\gamma(w)$ are two (row) vectors of 
dimension~$k$,
$\alpha(w)$ with entries in~$\{0,1\}$,
$\gamma(w)$ with entries in~$\N$.
The pair~$\Tran{\alpha(w),\gamma(w)}$ is computed by induction on the 
length of~$w$ in the following 
way.
Let~$\ell$ be the length of~$w$, 
let
\begin{equation}
\alpha(\unAe) = \lambda \EqVrgInt
\beta(\unAe) = \lambda \EqVrgInt\text{and} \e 
\gamma(\unAe) = 0
\eqvrg
\notag
\end{equation}
and, for every $0 \leq i < \ell$, let
\begin{multline}
\alpha(w_{i+1}) = \alpha(w_{i})\matmul\muf{a_{i+1}} \EqVrgInt
\beta(w_{i+1}) = \alpha(w_{i})\matmul\sigma_{a_{i+1}} \EqVrgInt\\
\text{and} \e 
\gamma(w_{i+1}) =  \lambda + \beta(w_{i+1})+\gamma(w_{i})\matmul\sigma
\eqpnt\e
\notag
% \ee
% \label{q.com-val-wor}
\end{multline}
% Finally, let
% $\msp\Tran{\alpha(w),\gamma(w)}=\Tran{\alpha_{\ell},\gamma_{\ell}}$.

% All~$\alpha_{i}$, and thus $\alpha(w)$, 
All~$\alpha(w)$ 
have entries in~$\{0,1\}$ 
since~$\Ac$ is unambiguous.
As a simple reformulation of the preceding subsection,
we have~$\msp \EvL{w} = \gamma(w)\matmul\nu \msp$ 
if~$\msp\alpha(w)\matmul\nu=1\msp$, that is,
if $w$ is recognised by~$\Ac$ and thus in~$L$, $\EvL{w}$ undefined 
otherwise.
This algorithm, that is, the computation 
of~$\Tran{\alpha(w),\gamma(w)}$, 
costs roughly~$6\xmd\ell\xmd k^2$ operations. 

\begin{example}
Let us consider again the language $L_1$ and the~{DFA}~$\Acu$ of
\figur{rec-num-1}. 
We have thus:
\begin{equation}
\sigma_a =
\begin{pmatrix}
    0 & 0 \\ 0 & 0
\end{pmatrix}
\EqVrgInt
\sigma_b =
\begin{pmatrix}
    1 & 0 \\ 0 & 1
\end{pmatrix}
\EqVrgInt
\sigma  =
\begin{pmatrix}
    1 & 1 \\ 1 & 1
\end{pmatrix}
\eqpnt 
\notag 
\end{equation}
The computation of~$\EvLu{bbabb}$ for instance takes the following 
steps.
\begin{equation*}
\begin{array}{c|c|c|c|c}
	i & a_i & \alpha _i & \beta _i & \gamma _i \\
	0 & & (1,0) & (1,0) & (0,0) \\
	1 & b & (0,1) & (1,0) & (2,0) \\
	2 & b & (1,0) & (0,1) & (3,3) %\\
% 	3 & a & (1,0) & (0,0) & (7,6) \\
% 	4 & b & (0,1) & (1,0) & (15,13) \\
% 	5 & b & (1,0) & (0,1) & (29,29) 
\end{array}
\eee
\begin{array}{c|c|c|c|c}
	i & a_i & \alpha _i & \beta _i & \gamma _i \\
% 	0 & & (1,0) & (1,0) & (0,0) \\
% 	1 & b & (0,1) & (1,0) & (2,0) \\
% 	2 & b & (1,0) & (0,1) & (3,3) \\
	3 & a & (1,0) & (0,0) & (7,6) \\
	4 & b & (0,1) & (1,0) & (15,13) \\
	5 & b & (1,0) & (0,1) & (29,29) 
\end{array}
\end{equation*}
And finally $\EvLu{bbabb} = (29,29)\matmul\nu_1 = 29$.
\end{example}

\begin{remark}
    The computation of~$\Tran{\alpha(w),\gamma(w)}$ is very similar 
    to the construction called \emph{product of an automaton by a 
    skew action} in~\cite{SakaSouz08,Souz08}.
\end{remark}

\section{Representation of recognisable subsets of numbers}

If~$s$ is an $\N$-rational series, that is, a map
$\msp s\colon\Ae\rightarrow\N\msp$,
it is well known that for any recognisable set of numbers~$X$,
$s^{-1}(X)$ is a rational set of~$\Ae$
(see~\cite[Corol. III.2.4]{BersReut88},
\cite[Th. VI.10.1]{Eile74} or~\cite[Corol. III,4,21]{Saka03}, 
for instance).
\theor{enu-ser-rat} thus directly implies the following statement,  
which has also been proved without reference to it 
in~\cite{LecoRigo01} and in~\cite{KrieEtAl09}. 

\begin{corollary}[\cite{LecoRigo01}]
\label{c.rec-num}%
A recognisable set of numbers is $L$-recognisable
in any rational abstract numeration system~$L$.
\end{corollary}

If~\corol{rec-num} requires formally no proof after the characterisation of 
rational abstract numeration systems given by~\theor{enu-ser-rat}, 
it is interesting to further  
investigate the construction which, given~$L$ and a recognisable set 
of numbers~$X$ computes an automaton which recognises the 
set~$\RL{X}$.
The computation method used in the preceding section (which is not 
the mere application of the general result that 
yields~\corol{rec-num}) allows to establish easily the following 
statement.

\begin{proposition}
    \label{p.rep-rec-sub}
Let~$L$ be a rational language over~$\Ae$ recognised by a 
\emph{deterministic} automaton of dimension~$k$.
For any integers~$p$ and~$r<p$, let~$X_{p,r}= p\N +r\msp$ be the set 
of integers congruent to~$r$ modulo~$p$.
Then the language~$\RL{X_{p,r}}$ of representations of numbers 
in~$X_{p,r}$ is recognised by a deterministic automaton with at 
most~$k\xmd p^{k}$ states.
\end{proposition}

\begin{proof}
Let~$\Ac$ be an automaton, with set of states~$Q$ of 
cardinal~$k$, which recognises~$L$ and~$\lmn$ the corresponding 
$\N$-representation.
If~$\Ac$ is deterministic, then~$\lambda$ and~$\mu$ are row monomial 
and so are all~$\alpha(w)$, for~$w$ in~$\Ae$, which 
are thus in 1-1 correspondence with the elements of~$Q$.

Let~$\Cc$ be the automaton whose set of states is
\begin{equation}
    R= \{\Tran{\alpha(w),\delta(w)}\mid w\in\Ae\}
    \ee\text{where}\ee
   \delta(w) = \gamma(w) \mod p
   \eqpnt 
    \notag
%     \label{}
\end{equation}
Thus,
$\msp R \subseteq Q \times (\Z/p\Z)^{k}\msp$.
The transitions of~$\Cc$ are defined by, for every~$a$ in~$A$:
\begin{equation}
    \fa w\in\Ae\quantvrg
    \fa a \in A \quantsp
    \Tran{\alpha(w),\delta(w)}\pathaut{a}{\Cc}\Tran{\alpha(wa),\delta(wa)}
    \eqpnt 
    \notag
%     \label{}
\end{equation}
The initial state of~$\Cc$ is~$(\lambda,0)$ and its final states are 
those~$\Tran{\alpha(w),\delta(w)}$ where~$\alpha(w)$ is final 
in~$\Ac$ and  $\msp\delta(w)\matmul\nu = r  \mod p\msp$.
% From~\equnm{com-val-wor} 
It then follows that the language accepted 
by~$\Cc$ is~$\RL{X_{p,r}}$.
\end{proof}

\begin{remark}
If we start from an unambiguous automaton~$\Ac$ of dimension~$k$, 
the same method yields a deterministic automaton~$\Cc$ with at 
most~$2^{k}p^{k}$ states.
\end{remark}

\begin{example}
The automaton built in this way from~$\Acu$ and for the 
recognisable set of numbers~$3\xmd\N +1$ is the automaton~$\Cc_1$ 
shown at
\figur{rec-num-2}
(this automaton is not minimal; its minimal quotient has only 
8~states).
\end{example}

\setlength{\lga}{3.8cm}\setlength{\lgb}{3cm}
\renewcommand{\ForthBackEdgeOffset}{3}
\begin{figure}[ht]
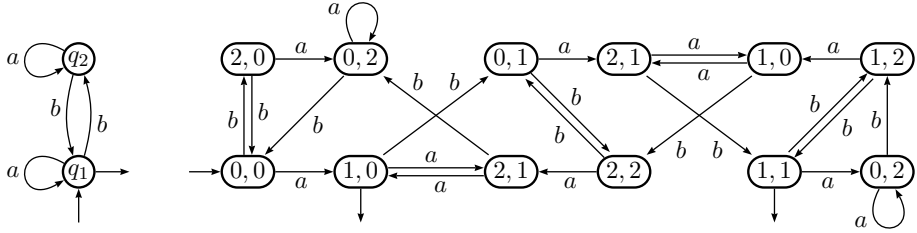

\centering
\VCDraw{%.5
\begin{VCPicture}{(-9,-1.4)(17,4.2)}
\VCPut{(-6.5,0)}{%    
\State[q_1]{(0,0)}{Q1}
\State[q_2]{(0,\lgb)}{Q2}
}% 
\Initial[s]{Q1}\Final{Q1}
\LoopW[.5]{Q1}{a}
\LoopW[.5]{Q2}{a}
\ArcR{Q1}{Q2}{b}
\ArcR{Q2}{Q1}{b}
\StateVar[0,0]{(-2,0)}{A}
\StateVar[2,0]{(-2,\lgb)}{B}
\StateVar[1,0]{(1,0)}{C}
\StateVar[0,2]{(1,\lgb)}{D}
\StateVar[2,1]{(5,0)}{E}
\StateVar[0,1]{(5,\lgb)}{F}
\StateVar[2,2]{(8,0)}{G}
\StateVar[2,1]{(8,\lgb)}{H}
\StateVar[0,2]{(15,0)}{I}
\StateVar[1,0]{(12,\lgb)}{J}
\StateVar[1,1]{(12,0)}{K}
\StateVar[1,2]{(15,\lgb)}{L}
\Initial{A} 
\Final[s]{K}\Final[s]{C}
\EdgeR[0.75]{E}{D}{b}
\EdgeL{D}{A}{b}
\EdgeL[0.75]{C}{F}{b}
\EdgeL[0.75]{J}{G}{b}
\EdgeR[0.75]{H}{K}{b}
\EdgeL{I}{L}{b}
\EdgeR{A}{C}{a}
\EdgeL{F}{H}{a}
\EdgeL{G}{E}{a}
\EdgeL{B}{D}{a}
\EdgeR{L}{J}{a}
\EdgeR{K}{I}{a}
\ForthBackOffset
\EdgeL{A}{B}{b}
\EdgeL{B}{A}{b}
\EdgeL{C}{E}{a}
\EdgeL{E}{C}{a}
\EdgeL{F}{G}{b}
\EdgeL{G}{F}{b}
\EdgeL{K}{L}{b}
\EdgeL{L}{K}{b}
\EdgeL{H}{J}{a}
\EdgeL{J}{H}{a}
\VarLoopOn%[.5][.5]
\LoopS{I}{a}
\LoopN{D}{a}
\end{VCPicture}}
\caption{A {DFA} recognising the set $3 \N +1$ in the ANS $\Lu$.}
        \label{f.rec-num-2}
\end{figure}

\begin{remark}
In~\cite{LecoRigo10}, another construction has been given for the 
same purpose. 
The automaton~$\Dc$ built with this other method and which 
recognises~$\RL{X_{p,r}}$ is not deterministic, but 
\emph{codeterministic} and has, roughly, $k\xmd p^{k+1}$~states.
Since~$\Dc$ is codeterministic, its determinisation yields the 
minimal automaton of~$\RL{X_{p,r}}$ and thus, thanks to 
\propo{rep-rec-sub}, does not produce an exponential blow-up.
We do not know of a direct proof of this fact.
\end{remark}

\begin{remark}
    \label{r.ref-dec}%
After the submitted version was written (and sent), we have  
learned of the reference~\cite{KrieEtAl09}.
Not only \corol{rec-num} is established there, but with a method of 
proof which is very similar to ours.
Our \lemme{enu-ser-rat} is Lemma~1 in~\cite{KrieEtAl09}.
The term \emph{representation} is not used there but the 
matrices~$\mu(a)$, $\sigma_{a}$ and~$\sigma$ are defined (under other 
notation) and used to give the same proof of \equat{pua-mu-1} 
(Lemma~2 in~\cite{KrieEtAl09}).

Afterwards, \cite{KrieEtAl09} develops in another direction than this 
paper: it proves lower bounds for the state complexity 
of~$\RL{X_{p,r}}$ and shows that the property corresponding to 
\corol{rec-num} does not hold for context-free languages. 
\end{remark}

\begin{remark}
If the numeration system considered is a positional numeration system 
(and still a rational one), and under some supplementary hypotheses, 
then the exact number of states for the minimal automaton 
of~$\RL{X_{p,r}}$ can be computed (\cf~\cite{CharEtAl10}).
\end{remark}

\section{Proof of \theor{enu-ser-dec}}

The image of a rational language by a rational relation (or 
transduction) is a rational language;
this classical result, due to Nivat and called Evaluation Theorem 
in~\cite{Eile74}, extends to rational series, as we state now 
(\cf~\cite{Saka03}).

\begin{proposition}
Let
$\msp\varphi\colon\Ae\rightarrow\Be\msp$
be an unambiguous rational relation 
and~$s$ a $\K$-rational series over~$\Ae$.
Then the series 
\begin{equation}
    \UL{\varphi}(s) = \sum_{w\in\Ae} \bra{s,w}\xmd \UL{\varphi}(w)
                  = \sum_{u\in\Be} \bra{s,\UL{\varphi^{-1}}(u)}\xmd u
    \eqvrg
    \label{q.ser-ima}
\end{equation}
if it is defined, is a $\K$-rational series over~$\Be$.
\end{proposition}

It is this result that was used in~\cite{ChofGold95} for the proof of 
\theor{enu-ser-rat}. 

\begin{proof}[Proof of \theor{enu-ser-dec}]
Let~$s$ in $\NRat\Ae$ and
$\msp L=\supp s\msp$ in~$\Rat\Ae$.
The set~$L$ is totally ordered by the radix order
\begin{equation}
    L = \{w_{0} < w_{1} < w_{2} < \ldots < w_{n} < \ldots \}
    \notag
%     \label{}
\end{equation}
and~$\SuccL$ is 
the function from~$\Ae$ into itself
whose domain is~$L$ and which maps every~$w_{i}$ to~$w_{i+1}$.
It is well-known that~$\SuccL$ is a rational function 
(\cite{BertEtAl07,FrouSak10}) and 
hence unambiguous (\cite{Eile74,Saka03}).
It then follows that the series
\begin{equation}
    \UL{\SuccL}(s) = \sum_{i=0}^{i=\infty} \bra{s,w_{i}}\xmd w_{i+1}
    \notag
%     \label{}
\end{equation}
is an $\N$-rational series and
$\msp t = s - \UL{\SuccL}(s)\msp$
is a $\Z$-rational series.
Now, $s$ is the enumerating series~$\EnL$ of the abstract numeration system~$L$ 
if, and only if, for every positive integer~$i$,
$\msp \bra{t,w_{i}} = 1\msp$, that is, if, and only if, 
$\msp t - \UL{L\bk\{w_{0}\}} = 0\msp$, a condition which is known to be decidable 
as~$\Z$ is a sub(semi)ring of a field (\cf \cite{Eile74,Saka03}).
\end{proof}

\section{Problems and future work}

Looking at abstract number systems as $\N$-rational series naturally
leads to two families of questions.
The first family consists in questions on $\N$-rational series which
ask to which extent the series is related to abstract number systems;
the second in questions which generalise to $\N$-rational series
questions that are usually considered for (abstract) numeration
systems.

An example of questions in the first family is to ask if it is
decidable whether a given $\N$-rational series is the enumeration (in
a radix ordering) of its (rational) support in a certain, and
unknown, abstract numeration system.
This seems to be rather a difficult problem.
An obvious necessary condition for a series to be a positive instance
of this problem is itself a non trivial problem that can be
formulated in the following way.

\begin{conjecture}
	It is decidable whether an $\N$-rational series is a monotone
	increasing function (for a given order of letters).
\end{conjecture}

A result due to Honkala \cite{Honk86} provides a kind of converse of
\corol{rec-num} in the case of $p$-ary numeration systems and states
that it decidable whether a $p$-recognisable set of numbers is
recognisable.
The generalisation of this result to larger class of
numeration systems has been recently studied 
in~\cite{BellEtAl09,Char09}.
Its generalisation to abstract number systems has been stated as a 
problem in~\cite{HonkRigo04}.
It is also a typical example of a question in the second family.

\begin{conjecture}
It is decidable whether the set of coefficients of an $\N$-rational
series is a recognisable set of numbers.
\end{conjecture}

\section{Summary}

In this short paper, we have presented a new idea for the study of 
abstract number systems, which brings to the subject the whole power 
of weighted automata theory.
% The first results show that this theory is fully relevant.
In return, the subject of abstract number systems naturally opens new 
questions for the theory of $\N$-rational series.

\section*{Acknowledgements}
The authors are grateful to Sylvain Lombardy for the discussions with him on
this problem. They also have greatly benefited from having access to the
chapter of Lecomte and Rigo on abstract numeration systems in the book 
\emph{Combinatorics, {A}utomata and {N}umber {T}heory}
to be published by Cambridge University Press.

The authors are particularly indebted to Michel Rigo who, after 
seeing a first version of this paper, gave them the information that 
Theorem~1 was already known, the pertinent references, and the 
encouragements to nevertheless write down a complete version.

\end{document}